\documentclass[a4paper,11pt]{article}

\usepackage{adjustbox}
\usepackage{amsfonts}
\usepackage{amsmath}
\usepackage{amsthm}
\usepackage{bm}
\usepackage{algorithm}
\usepackage{algpseudocode}
\usepackage{graphicx}
\usepackage{tabularx}

\usepackage{geometry}
\newtheorem{theorem}{Theorem}[section]
\newtheorem{corollary}{Corollary}[theorem]
\newtheorem{lemma}[theorem]{Lemma}

\usepackage[colorlinks,bookmarksopen,bookmarksnumbered,citecolor=red,urlcolor=red]{hyperref}

\newcommand\BibTeX{{\rmfamily B\kern-.05em \textsc{i\kern-.025em b}\kern-.08em
T\kern-.1667em\lower.7ex\hbox{E}\kern-.125emX}}

\newcommand\barbelow[1]{\stackunder[1.2pt]{$#1$}{\rule{.8ex}{.075ex}}}
\renewcommand{\bar}{\overline} 
\renewcommand{\barbelow}{\underline} 

\newcommand{\footremember}[2]{
    \footnote{#2}
    \newcounter{#1}
    \setcounter{#1}{\value{footnote}}
}
\newcommand{\footrecall}[1]{
    \footnotemark[\value{#1}]
}

\providecommand{\keywords}[1]
{
  \small	
  \textbf{\textit{Keywords---}} #1
}

\title{Posterior Uncertainty Estimation via a Monte Carlo Procedure Specialized for Data Assimilation}

\author{
Michael Stanley\footremember{cmustats}{Department of Statistics and Data Science, Carnegie Mellon University, Pittsburgh, PA, USA.} 
\\ {\small \href{mailto:mcstanle@andrew.cmu.edu}{mcstanle@andrew.cmu.edu}~} 
\and Mikael Kuusela\footrecall{cmustats}
\\ {\small \href{mailto:mkuusela@andrew.cmu.edu}{mkuusela@andrew.cmu.edu}~} 
\and Brendan Byrne\footremember{jpl}{Jet Propulsion Laboratory, California Institute of Technology, Pasadena, CA, USA.} 
\\ {\small \href{mailto:brendan.k.byrne@jpl.nasa.gov}{brendan.k.byrne@jpl.nasa.gov}~} 
\and Junjie Liu\footrecall{jpl} 
\\ {\small \href{mailto:junjie.liu@jpl.nasa.gov}{junjie.liu@jpl.nasa.gov}~} 
}

\begin{document}
\date{}
\maketitle

\begin{abstract}
    Through the Bayesian lens of data assimilation, uncertainty on model parameters is traditionally quantified through the posterior covariance matrix. However, in modern settings involving high-dimensional and computationally expensive forward models, posterior covariance knowledge must be relaxed to deterministic or stochastic approximations. In the carbon flux inversion literature, Chevallier et al. \cite{chevallier} proposed a stochastic method capable of approximating posterior variances of linear functionals of the model parameters that is particularly well-suited for large-scale Earth-system data assimilation tasks. This note formalizes this algorithm and clarifies its properties. We provide a formal statement of the algorithm, demonstrate why it converges to the desired posterior variance quantity of interest, and provide additional uncertainty quantification allowing incorporation of the Monte Carlo sampling uncertainty into the method's Bayesian credible intervals. The methodology is demonstrated using toy simulations and a realistic carbon flux inversion observing system simulation experiment.
\end{abstract}


\section{Introduction} \label{sec:intro_new}
Uncertainty quantification (UQ) for data assimilation (DA) tasks is often non-trivial, but scientifically paramount to their understanding and interpretation. Since DA broadly describes methods combining observations with a computational model of a physical system, a Bayesian framework is often sensible for inference on the model parameters, as the posterior distribution quantifies knowledge resulting from this combination. As such, Bayesian statistical models are regularly used as the UQ framework. For example, Bayesian procedures play a central role in the general idea of optimal estimation \cite{rodgers}, the broad field of DA \cite{kalnay}, and the more specific field of carbon flux estimation \cite{deng_2014, liu_2016}. Inference for DA tasks using this statistical framework is typically challenging due to high-dimensional settings (e.g., high-resolution spatiotemporal grids) and the computer model’s implicit numerical definition of the physical system of interest, often requiring supercomputers and long compute times. Prior and observation error distributions are often assumed to be Gaussian, yielding a Gaussian posterior distribution under a linear forward model. Although a Gaussian posterior can be exactly characterized by its mean vector and covariance matrix, the high-dimensionality makes dealing directly with the posterior covariance matrix intractable and the implicit computationally demanding forward model makes infeasible standard traditional Bayesian computational techniques, such as Markov Chain Monte Carlo (MCMC). The implicit posterior necessitates the development of computational methods that implicitly access it.

CO$_2$ flux inversion is a representative example of a high-dimensional DA task to which Bayesian modeling is applied and used to compute estimated flux fields \cite{Crowell-2019, enting1995synthesis,gurney2002towards}.
In this problem, estimates of net surface-atmosphere CO$_2$ fluxes are inferred from atmospheric CO$_2$ measurements, with fluxes and atmospheric measurements being related by a chemical transport model (the computational forward model). However, the relatively sparse atmospheric CO$_2$ observations underconstrain surface fluxes of CO$_2$, and regularization with prior information is the Bayesian approach to making the problem well-posed. These analyses have historically assimilated measurements of atmospheric CO$_2$ from a global network of flask and in situ measurements \cite{enting1995synthesis}, but more recent work \cite{Byrne-2023,Crowell-2019, deng_2014,Houweling-2015} has shifted to assimilating space-based column-averaged dry-air mole fractions, denoted $\mathrm{X_{CO_2}}$, as observations availability has expanded since 2009. In these analyses, the prior and error distributions are typically assumed to be Gaussian and the forward model can be reasonably assumed linear in the net surface-atmosphere fluxes.

When the number of model parameters is low and a forward model run is inexpensive, it is possible to explicitly construct the posterior covariance matrix. Successful examples of this approach date back at least to Vasco et al.\ \cite{vasco} in seismic tomography, where inversion is performed on $12{,}496$ model parameters. However, more contemporary problems typically have orders of magnitude more parameters and substantially more expensive forward models, requiring other approaches to access posterior covariance matrix information. Once the discretization of the computational model is set, the dimensionality problem can be handled either by defining an approximate statistical model on a lower dimensional problem, or by working in some subspace of the full-dimensional problem. A recent example of the first strategy is seen in Zammit-Mangion et al.\ \cite{wombat} in the WOMBAT inversion system which lowers the dimension of the statistical model via an intelligently chosen set of basis functions, facilitating MCMC. Alternatively, Petra et al.\ \cite{petra_mcmc} propose with Stochastic Newton MCMC (SN-MCMC) the possibility for MCMC in the full parameter space by using a low-rank approximation to the posterior covariance within the proposal distribution of a Metropolis--Hasting algorithm. Although WOMBAT and SN-MCMC are both MCMC-based, WOMBAT assumes a linear forward model, while SN-MCMC does not, allowing it to characterize non-Gaussian posteriors. Staying with a linear forward model assumption, other approaches leverage low-rank posterior covariance approximations. Flath et al.\ \cite{flath} develop a low-rank algorithm for approximating the posterior covariance by computing the leading eigenvalues and eigenvectors of a prior-conditioned Hessian matrix of the associated objective function (i.e., the log posterior). In a similar spirit, Kalmikov and Heimbach \cite{kalmikov} provide a derivative-based algorithm to compute leading Hessian eigenvalues and eigenvectors and extend the uncertainty quantification to quantities of interest in global ocean state estimation. The algorithms in both Flath et al.\ \cite{flath} and Kalmikov and Heimbach \cite{kalmikov} rely upon the Lanczos method \cite{lanczos} for matrix-free computation of the low-rank approximation. Alternatively, Bousserez and Henze \cite{bousserez_henze} more recently proposed a low-rank approximation algorithm dependent upon the randomized SVD algorithm \cite{halko}. All of the aforementioned methods can be grouped by their reliance upon some low-dimensional deterministic approximation.

In contrast, stochastic approximations of the posterior distribution rely neither upon pre-inversion dimension reductions nor low-rank matrix approximations, but rather generate ensembles of inversions using random generators. In carbon flux inversion, Chevallier et al.\ \cite{chevallier} developed such a method to estimate the posterior variance of \emph{functionals} of the flux field (i.e., maps from the flux field to the reals). The method uses the forward model, specified prior, and known observation error distributions in a particularly efficient manner. Broadly, the algorithm creates an ensemble of prior means and observation errors, sampling according to their respective distributions. For each ensemble member, it finds the maximum a posteriori (MAP) estimator, to which the functional is applied. Finally, it finds the empirical variance across the ensemble members to estimate the posterior variance of the functional. This method is well-suited for carbon flux estimation and DA UQ more generally for a few key reasons. First, each ensemble member is computationally independent, making the method parallelizable and hence offering a substantial computational benefit compared to sequential methods, such as MCMC. Second, although in general prior misspecification biases the posterior, the prior mean does not need to be correctly specified in order for the procedure to produce an unbiased estimator of the posterior variance. Third, the ensemble of inversions can flexibly produce UQ estimates for arbitrary functionals post hoc, as opposed to requiring the specification of a functional ahead of the analysis. Finally, since this method is more generally a Monte Carlo (MC) method for a Gaussian statistical model, the method's sampling uncertainty can be analytically characterized and accounted for in the final UQ estimate. The ability to easily characterize this uncertainty of the uncertainty stands in contrast to the difficulty in characterizing deterministic error of the aforementioned low-dimensional approaches.

Although Chevallier et al.\ \cite{chevallier} appear to have been the first to develop this method, which was later applied in Liu et al.\ \cite{liu_2016}, we are unaware of a formal statement or analysis of this algorithm. These previous works also did not quantify the algorithm's MC uncertainty. As such, the primary contributions of this paper are a rigorous formal statement of the algorithm, an analysis showing the convergence of its output to the true posterior quantity of interest, and uncertainty quantification of the algorithm itself so that the algorithm's sampling uncertainty can be accounted for in the final inference.

The rest of this paper is structured as follows. In Section~\ref{sec:method}, we fully describe the algorithm, present mathematical results proving its correctness, and derive deflation and inflation factors to apply to the estimated posterior uncertainty to quantify the MC uncertainty. Proofs of the mathematical results can be found in Appendix~\ref{sec:appendix_multivariate_algebra}. In Section~\ref{sec:experimental_examples}, we provide two experimental demonstrations: the first is a low-dimensional problem in which we explicitly know the linear forward model and the second is a carbon flux observing system simulation experiment (OSSE) to which we applied this method to compute global monthly flux credible intervals along their MC uncertainty. Finally, we provide some concluding remarks in Section~\ref{sec:conclusion}. For reference, all mathematical notation in order of appearance is collected in Table~\ref{table:math_symbols}.

\begin{table}[t]
    \caption{Mathematical symbols and notation used herein (in order of appearance).}
    \label{table:math_symbols}
    \centering
    \begin{adjustbox}{max width=\textwidth}{
    \begin{tabular}{l l | l l }
        \hline
        $\mathbf{c} \in \mathbb{R}^m$ & Scaling factors & $\tilde{\mathbf{y}} \in \mathbb{R}^n$ & X$_{CO_2}$ Observations \\
        $\bm{\mu} \in \mathbb{R}^m$ & Control fluxes & $f$ & Forward model \\
        $\bm{\epsilon} \in \mathbb{R}^n$ & Observation Noise & $\mathbf{R}$ & Observation Noise Covariance \\
        $\mathbf{A} \in \mathbb{R}^{n \times m}$ & Linear forward model & $\mathbf{a} \circ \mathbf{b}$ & Element-wise multiplication of $\mathbf{a}$ and $\mathbf{b}$ \\
        $\mathbf{z} \in \mathbb{R}^n$ & Non-biospheric X$_{CO_2}$ component & $\mathbf{y} \in \mathbb{R}^n$ & Biospheric component of X$_{CO_2}$ \\
        $\mathbf{c}^b \in \mathbb{R}^m$ & Prior scaling factor expectation & $\mathbf{B} \in \mathbb{R}^{m \times m}$ & Prior scaling factor covariance \\
        $\pi(\mathbf{c} \mid \mathbf{y})$ & Posterior scaling factor density & $b^2 \in \mathbb{R}_+$ & Prior variance parameter \\
        $\mathbf{I}_m \in \mathbb{R}^{m \times m}$ & $m \times m$ identity matrix & $\bm{\theta} = \mathbf{c} \circ \bm{\mu}$ & Flux vector \\
        $\mathbf{A}_{\bm{\mu}} \in \mathbb{R}^{n \times m}$ & Forward model with control flux $\bm{\mu}$ & $\mathbf{\Sigma} \in \mathbb{R}^{m \times m}$ & Posterior scaling factor covariance \\
        $\bm{\alpha} \in \mathbb{R}^m$ & Posterior scaling factor expectation & $\bm{\delta} \in \mathbb{R}^m$ & Posterior flux expectation \\
        $\bm{\Gamma} \in \mathbb{R}^{m \times m}$ & Posterior flux covariance & $(\mathbf{c}_k, \mathbf{y}_k) \in \mathbb{R}^m \times \mathbb{R}^n$ & $k$th MC sample \\
        $\bm{\Sigma}_{\mathbf{c}_{MAP}^k} \in \mathbb{R}^{m \times m}$ & MC MAP estimator covariance & $\mathbf{h} \in \mathbb{R}^m$ & Functional of interest \\
        $\bar{\varphi} \in \mathbb{R}$ & Mean MC functional value & $\hat{\sigma}^2_\varphi \in \mathbb{R}_+$ & Empirical functional variance \\
        $\chi^2_{M - 1}$ & Chi-squared distribution with $M - 1$ dof & $\chi_{M - 1, \alpha / 2}^2$ & Chi-squared $(\alpha / 2)$-quantile \\
        $\alpha \in (0, 1)$ & Frequentist confidence level & $\gamma \in (0, 1)$ & Bayesian credible interval level \\
        $L$ & Deflation factor for MC variance & $R$ & Inflation factor for MC variance \\
        \hline
    \end{tabular}}
    \end{adjustbox}
\end{table}

\section{Monte Carlo Method Exposition, Analysis, \\and Uncertainty Quantification} \label{sec:method}
\subsection{The Bayesian 4D-Var Setup} \label{subsec_bayes_setup}
Following along with the mathematical setup of Henze et al. \cite{Henze-2007}, the prior and posterior distributions are defined in a scaling factor space and hence the prior and posterior distributions on the physical quantity of interest are obtained by multiplying the respective scaling factor by a control quantity. In carbon flux estimation, the control quantity is a control flux, typically an ansatz CO$_2$ flux between the Earth's surface and the atmosphere. Note that if the prior distribution mean in scaling factor space is unity, then the control flux is also the prior mean in the physical quantity of interest space. Mathematically, fix the scaling factor vector $\mathbf{c} \in \mathbb{R}^m$ and let $\tilde{\mathbf{y}} \in \mathbb{R}^n$ be the observation vector and $\bm{\mu} \in \mathbb{R}^m$ the control physical quantity. The following model generally describes the relationship between the scaling factors and the observations:
\begin{equation}
    \tilde{\mathbf{y}} = f(\mathbf{c}; \bm{\mu}) + \bm{\epsilon}, \quad \bm{\epsilon} \sim \mathcal{N}(\bm{0}, \mathbf{R}),
\end{equation}
where $\mathbf{R} \in \mathbb{R}^{n \times n}$ is the observation covariance matrix. The observation vector $\tilde{\mathbf{y}}$ is a sequence of X$_{CO_2}$ observations produced from a remote sensing satellite, e.g., GOSAT or OCO-2 \cite{odell}. This model expression is a composition of an atmospheric transport model mapping scaling factors to atmospheric CO$_2$ concentrations composed with a remote sensing observation operator mapping CO$_2$ concentrations to X$_{CO_2}$ scalar values. The atmospheric transport model is known to be affine due to the physics of CO$_2$ atmospheric transport. In reality, the true mapping from atmospheric CO$_2$ concentrations to X$_{CO_2}$ is non-linear, but in line with \cite{liu_2016}, we use an affine form involving the known GOSAT averaging kernel. As such, the affine composed function $f$ is of the form $f(\mathbf{c}; \bm{\mu}) = \mathbf{A} (\mathbf{c} \circ \bm{\mu}) + \mathbf{z}$, where $\mathbf{A} \in \mathbb{R}^{n \times m}$ is the linear forward model matrix, $\mathbf{c} \circ \bm{\mu}$ denotes the component-wise scaling of $\bm{\mu}$ by the scaling factors $\mathbf{c}$, and $\mathbf{z}$ is comprised of the non-biospheric CO$_2$ contribution to the observations along with the prior mean of the X$_{CO_2}$ retrieval algorithm. As such, we define $\mathbf{y} := \tilde{\mathbf{y}} - \mathbf{z}$, giving the linear model
\begin{equation} \label{eq:linear_forward_model}
    \mathbf{y} = \mathbf{A} (\mathbf{c} \circ \bm{\mu}) + \bm{\epsilon}, \quad \bm{\epsilon} \sim \mathcal{N}(\bm{0}, \mathbf{R}),
\end{equation}
on which the following analysis is performed. For ease of reference, we herein refer to the linear map $\mathbf{A}$ as the forward model. We emphasize that the matrix $\mathbf{A}$ is not explicitly available to us, but is implicitly defined by the atmospheric transport model (and the satellite observation operator, which is usually explicitly available).

To regularize the problem and provide uncertainty quantification on the estimated scaling factors and fluxes, $\mathbf{c}$ in Equation~\eqref{eq:linear_forward_model} is given a Gaussian prior distribution, yielding the following Bayesian generative model:
\begin{align}
    \mathbf{c} &\sim \mathcal{N}(\mathbf{c}^b, \mathbf{B}) \label{eq:prior_c}, \\
    \mathbf{y} \mid \mathbf{c} &\sim \mathcal{N}(\mathbf{A} (\mathbf{c} \circ \bm{\mu}), \mathbf{R}) \label{eq:like_y},
\end{align}
where $\mathbf{c}^b \in \mathbb{R}^m$ is the scaling factor prior mean and $\mathbf{B}$ is the prior covariance matrix.

Finding the posterior mean (or, equivalently, the posterior mode) defined by Equations~\eqref{eq:prior_c} and \eqref{eq:like_y} characterizes a common DA problem of interest. A typical DA approach in carbon flux estimation is four dimensional variational data assimilation (4D-Var) \cite{deng_2014, liu_2016}, a method optimizing carbon fluxes simultaneously over all time steps. 4D-Var can be regarded as a least-squares optimization with an $\ell_2$ regularizer (ridge regression), or, equivalently, as maximum a posteriori (MAP) estimation in the Bayesian paradigm. This connection means that the 4D-Var optimization is connected to the posterior resulting from the prior and likelihood in Equations \eqref{eq:prior_c} and \eqref{eq:like_y}. From the Bayesian perspective, the 4D-Var cost function $F(\mathbf{c})$ is the negative log-posterior density of the scaling factors given the observations,
\begin{equation} \label{eq:log_posterior}
\begin{split}
    F(\mathbf{c}) = -\log (\pi(\mathbf{c} \mid \mathbf{y})) = \frac{1}{2} \left( \mathbf{c} - \mathbf{c}^b \right)^{\top} \mathbf{B}^{-1} \left( \mathbf{c} - \mathbf{c}^b \right) \\
    + \frac{1}{2} \left( \mathbf{y} - \mathbf{A}(\mathbf{c} \circ \bm{\mu}) \right)^{\top} \mathbf{R}^{-1} \left( \mathbf{y} - \mathbf{A}(\mathbf{c} \circ \bm{\mu}) \right) + C,
\end{split}
\end{equation}
where $C \in \mathbb{R}$ is a normalizing constant for the posterior distribution and $\pi(\mathbf{c} \mid \mathbf{y})$ denotes the posterior density. Thus, finding the MAP estimator, i.e., the $\mathbf{c}$ that maximizes the posterior density, is equivalent to finding the vector $\mathbf{c}$ that minimizes the 4D-Var cost function.

In this study (Sect.~\ref{sec:osse_example}), the prior covariance $\mathbf{B}$ is parameterized with a single real value, $\mathbf{B} := b^2 \mathbf{I}_m$, where $b \in \mathbb{R}$. This is inline with several published studies \cite{deng_2014, liu_2016}, and implies that all prior spatio-temporal indices are statistically independent. Similarly, the noise covariance $\mathbf{R}$ is assumed to be a diagonal matrix where each diagonal element is simply the variance of the corresponding X$_{CO_2}$ observation. As such, each diagonal element depends on the uncertainty of its corresponding X$_{CO_2}$ retrieval and the observations are assumed statistically independent given the scaling factors.

If the forward model $\mathbf{A}$ is known explicitly, the posterior mean and covariance of $\mathbf{c}$ are analytically tractable. We can then find the posterior uncertainty of the physical quantity $\bm{\theta} = \mathbf{c} \circ \bm{\mu}$ given the observations. We rewrite Equation \eqref{eq:like_y} using the short-hand notation for the $\circ$ operation (see Appendix~\ref{sec:appendix_multivariate_algebra}), i.e., $\mathbf{A} (\mathbf{c} \circ \bm{\mu}) = \mathbf{A}_{\bm{\mu}}\mathbf{c}$:
\begin{equation} \label{eq:like_y_circ}
    \mathbf{y} \mid \mathbf{c} \sim \mathcal{N}(\mathbf{A}_{\bm{\mu}} \mathbf{c}, \mathbf{R}).
\end{equation}
Hence, $\mathbf{c} \mid \mathbf{y} \sim \mathcal{N}(\bm{\alpha}, \bm{\Sigma})$, where by the posterior mean and covariance Equations 4.3 and 4.7 in \cite{rodgers} we have:
\begin{align}
    \bm{\Sigma} &= \left( \frac{1}{b^2} \mathbf{I}_m +  \mathbf{A}_{\bm{\mu}}^{\top} \mathbf{R}^{-1} \mathbf{A}_{\bm{\mu}} \right)^{-1} = \left( \left(\mathbf{A}^{\top} \mathbf{R}^{-1} \mathbf{A} \right) \circ \bm{\mu} \bm{\mu}^{\top} + \frac{1}{b^2} \mathbf{I}_m \right)^{-1}, \label{eq:sf_posterior_cov} \\
    \bm{\alpha} &= \bm{\Sigma} \left( \mathbf{A}_{\bm{\mu}}^{\top} \mathbf{R}^{-1} \mathbf{y} + \frac{1}{b^2} \mathbf{c}^b \right) = \bm{\Sigma} \left( \left( \mathbf{A}^{\top} \mathbf{R}^{-1} \mathbf{y} \right) \circ \bm{\mu} + \frac{1}{b^2} \mathbf{c}^b \right), \label{eq:MAP_est_sf}
\end{align}
where Equation~\eqref{eq:sf_posterior_cov} follows from Corollary~\ref{cor:psd} and Equation~\eqref{eq:MAP_est_sf} follows from Lemma \ref{lem:circ_through_matmul}. Note, $\bm{\alpha}$ is also the MAP estimator of $\mathbf{c}$. Furthermore, the posterior distribution for the physical quantity, $\bm{\theta} = \mathbf{c} \circ \bm{\mu}$, is $\bm{\theta} \mid \mathbf{y} \sim \mathcal{N}(\bm{\delta}, \bm{\Gamma})$, where
\begin{align}
    \bm{\delta} &= \mathbb{E}[\bm{\theta} \mid \mathbf{y}] = \mathbb{E}[\mathbf{c} \circ \bm{\mu} \mid \mathbf{y}] = \mathbb{E}[\mathbf{c} \mid \mathbf{y}] \circ \bm{\mu} = \bm{\alpha} \circ \bm{\mu}, \label{eq:flux_posterior_mean} \\
    \bm{\Gamma} &= \text{Cov}[ \bm{\theta} \mid \mathbf{y}] = \text{Cov}[\mathbf{c} \circ \bm{\mu} \mid \mathbf{y}] = \text{Cov}[\mathbf{c} \mid \mathbf{y}] \circ \bm{\mu} \bm{\mu}^{\top} = \bm{\Sigma} \circ \bm{\mu} \bm{\mu}^{\top}, \label{eq:flux_posterior_cov}
\end{align}
where we have used Lemmas~\ref{lemma:E_circ} and \ref{lemma:cov_circ} from Appendix~\ref{sec:appendix_multivariate_algebra}.

In practice, the forward model $\mathbf{A}$ is only known implicitly via a computer simulator, making direct use of Equations~\eqref{eq:sf_posterior_cov}--\eqref{eq:flux_posterior_cov} intractable. Instead, the 4D-Var cost function in Equation~\eqref{eq:log_posterior} is minimized using the L-BFGS-B algorithm \cite{byrd_lbfgsb} with the cost function gradient computed numerically using the adjoint method \cite{Henze-2007}. After a handful of iterations, L-BFGS-B finds a reasonable approximation of the posterior mean~/~mode $\bm{\alpha}$, which yields a point estimator of $\bm{\theta}$ using Equation~\eqref{eq:flux_posterior_mean}. We now describe a procedure that provides an approach for uncertainty quantification of $\bm{\theta}$ despite the intractability of the posterior covariance $\bm{\Sigma}$.

\subsection{The Monte Carlo Procedure} \label{subsec_monte_carlo_proc}
To execute the Monte Carlo procedure introduced in \cite{chevallier}, we generate $M$ ensemble members. For each $k = 1, 2, \dots, M$, we sample a new prior mean $\mathbf{c}_k$ and new observation $\mathbf{y}_k$ as follows:
\begin{align}
    \mathbf{c}_k &\overset{\text{i.i.d.}}{\sim} \mathcal{N}(\mathbf{1}, b^2 \mathbf{I}_m) \label{eq:mc_prior}, \\
    \mathbf{y}_k &\overset{\text{i.i.d.}}{\sim} \mathcal{N}(\mathbf{A} \bm{\mu}, \mathbf{R}), \label{eq:mc_likelihood}
\end{align}
where $b^2$ is the prior uncertainty mentioned in Section~\ref{subsec_bayes_setup}. Notice that $\mathbf{A}\bm{\mu}$ is known after a single forward model run, and hence Equation~\eqref{eq:mc_likelihood} is more illuminatingly seen as sampling Gaussian noise for each Monte Carlo sample: letting $\bm{\epsilon}_k \sim \mathcal{N}(\bm{0}, \mathbf{R})$, for $k = 1, \dots, M$, each Monte Carlo iteration involves sampling a pair, $(\mathbf{c}_k, \mathbf{y}_k ) \in \mathbb{R}^{m} \times \mathbb{R}^{n}$, where $\mathbf{y}_k := \mathbf{A} \bm{\mu} + \bm{\epsilon}_k$ for each $k = 1, \dots, M$.

The MAP estimator from Equation~\eqref{eq:MAP_est_sf} corresponding to prior mean $\mathbf{c}_k$ and observation $\mathbf{y}_k$ is analytically tractable for each ensemble member when $\mathbf{A}$ is explicitly known:
\begin{equation}
    \mathbf{c}_{MAP}^k = \bm{\Sigma} \left( \left( \mathbf{A}^{\top} \mathbf{R}^{-1} \mathbf{y}_k \right) \circ \bm{\mu} + \frac{1}{b^2} \mathbf{c}_k \right). \label{eq:sf_map_est}
\end{equation}
\noindent Similarly, the covariance matrix of this MAP estimator, henceforth denoted as $\bm{\Sigma}_{\mathbf{c}_{MAP}^k}$, is
\begin{align}
    \bm{\Sigma}_{\mathbf{c}_{MAP}^k} &= \bm{\Sigma} \: \text{Cov} \left[\mathbf{A}^{\top} \mathbf{R}^{-1} \mathbf{y}_k \circ \bm{\mu} + \frac{1}{b^2} \mathbf{c}_k \right] \bm{\Sigma}^\top \label{eq:map_covariance} \\
    &= \bm{\Sigma} \left(\text{Cov} \left[ \mathbf{A}^{\top} \mathbf{R}^{-1} \mathbf{y}_k \circ \bm{\mu} \right] + \frac{b^2}{b^4} \mathbf{I}_m \right) \bm{\Sigma} \nonumber \\
    &= \bm{\Sigma} \left(\text{Cov} \left[ \mathbf{A}^{\top} \mathbf{R}^{-1} \mathbf{y}_k \right] \circ \bm{\mu} \bm{\mu}^{\top} + \frac{1}{b^2} \mathbf{I}_m \right) \bm{\Sigma} \nonumber \\
    &= \bm{\Sigma} \left( \left( \mathbf{A}^{\top} \mathbf{R}^{-1} \text{Cov} \left[ \mathbf{y}_k \right] \mathbf{R}^{-1} \mathbf{A} \right) \circ \bm{\mu} \bm{\mu}^{\top} + \frac{1}{b^2} \mathbf{I}_m \right) \bm{\Sigma} \nonumber \\
    &= \bm{\Sigma} \left( \left( \mathbf{A}^{\top} \mathbf{R}^{-1} \mathbf{A} \right) \circ \bm{\mu} \bm{\mu}^{\top} + \frac{1}{b^2} \mathbf{I}_m \right) \bm{\Sigma} = \bm{\Sigma}, \nonumber  \label{eq:map_est_cov}
\end{align}
since $\bm{\Sigma}^{-1} = \left( \mathbf{A}^{\top} \mathbf{R}^{-1} \mathbf{A} \circ \bm{\mu} \bm{\mu}^{\top} + \frac{1}{b^2} \mathbf{I}_m \right)$. This shows that the covariance matrix of the Monte Carlo ensemble of scaling factor MAP estimators is equal to the desired posterior covariance $\bm{\Sigma}$. Note, the $\circ$ operation step on the third line of derivation~\eqref{eq:map_covariance} follows from Lemma~\ref{lemma:cov_circ} in Appendix~\ref{sec:appendix_multivariate_algebra}. To the best of our knowledge, proof of this equality has not appeared in previous literature on this method.

This covariance equality also exists in the physical quantity space, e.g., carbon flux space. The estimator of the physical quantity corresponding to $\mathbf{c}_{MAP}^k$ is
\begin{equation}
    \bm{\theta}_k = \mathbf{c}_{MAP}^k \circ \bm{\mu}.
\end{equation}
Using the result from Lemma~\ref{lemma:cov_circ} in Appendix~\ref{sec:appendix_multivariate_algebra}, the covariance matrix $\text{Cov}[\bm{\theta}_k]$ of this estimator is
\begin{equation}
    \text{Cov}[\bm{\theta}_k] = \text{Cov}[\mathbf{c}_{MAP}^k \circ \bm{\mu}] = \text{Cov}[\mathbf{c}_{MAP}^k] \circ \bm{\mu} \bm{\mu}^{\top} = \bm{\Sigma} \circ \bm{\mu} \bm{\mu}^{\top} = \bm{\Gamma}. \label{eq:mc_cov_phys_quant}
\end{equation}
Hence, the covariance matrix of the Monte Carlo physical quantity estimator is equal to the posterior covariance matrix of that physical quantity.

However, for most DA tasks the forward model is not explicitly available, so each ensemble member MAP estimator $\mathbf{c}_{MAP}^k$ must be obtained with an iterative optimization algorithm minimizing \eqref{eq:log_posterior} with $\mathbf{c}_k$ and $\mathbf{y}_k$ as the prior mean and observation vectors. Once these ensemble members are obtained, we could in principle estimate the posterior scaling factor covariance matrix $\bm{\Sigma}$ with the empirical covariance estimator $\widehat{\bm{\Sigma}}$ based on the Monte Carlo ensemble as follows,
\begin{equation} \label{eq:empirical_cov}
    \widehat{\bm{\Sigma}} = \frac{1}{M - 1} \sum_{k = 1}^M \left(\mathbf{c}^k_{MAP} - \bar{\mathbf{c}} \right)\left(\mathbf{c}^k_{MAP} - \bar{\mathbf{c}} \right)^{\top},
\end{equation}
where $\bar{\mathbf{c}} = \frac{1}{M} \sum_{k = 1}^M \mathbf{c}^k_{MAP}$. To translate Equation~\eqref{eq:empirical_cov} to the physical quantity space, we simply plug in the empirical covariance estimator to Equation~\eqref{eq:mc_cov_phys_quant}, i.e.,
\begin{equation}
    \hat{\bm{\Gamma}} = \widehat{\text{Cov}[\bm{\theta}_k]} = \widehat{\bm{\Sigma}} \circ \bm{\mu} \bm{\mu}^\top.
\end{equation}

In practice, DA scenarios like carbon flux inversion are typically high dimensional, making direct interaction with these covariance matrices difficult. Indeed, accurate estimation of $\bm{\Sigma}$ using Equation~\eqref{eq:empirical_cov} would require an enormously large Monte Carlo ensemble and would require storing and working with an $m \times m$ matrix, where $m \sim 10^5$ or larger. Fortunately, we often care about the variance of one-dimensional summaries of $\bm{\theta}$, such as the posterior flux variance for a specific region during some time period, as opposed to the full posterior covariance matrix. For instance, we might wish to estimate North American fluxes over some month. For the remaining presentation of these ideas, we only discuss the posterior scaling factor vector, $\mathbf{c} \mid \mathbf{y}$, but as we have shown above, obtaining the posterior flux is achieved by a simple component-wise scaling with the control flux $\bm{\mu}$.

Obtaining quantities of the above type is mathematically implemented using a linear functional of the underlying high-dimensional parameter. That is, we wish to characterize the posterior of $\varphi(\mathbf{c}) = \bm{h}^{\top} \mathbf{c}$, where $\bm{h} \in \mathbb{R}^m$ contains weights necessary to aggregate the desired scaling factors. Hence, building off Equations~\eqref{eq:MAP_est_sf} and \eqref{eq:sf_posterior_cov}, we obtain the posterior distribution for the functional of interest:
\begin{equation}
    \varphi(\mathbf{c}) \mid \mathbf{y} \sim \mathcal{N}(\bm{h}^{\top} \bm{\alpha}, \bm{h}^{\top} \bm{\Sigma} \bm{h}).
\end{equation}
We wish to obtain the posterior variance of this functional. Define $\sigma^2_{\varphi} = \text{Var}(\varphi(\mathbf{c}) \mid \mathbf{y}) = \bm{h}^{\top} \bm{\Sigma} \bm{h}$. We could inefficiently estimate this using $\hat{\sigma}^2_{\varphi} = \bm{h}^{\top} \hat{\bm{\Sigma}} \bm{h}$, but we wish to avoid working directly with the full empirical covariance matrix. The following algebraic steps provide a better alternative:
\begin{align}
    \hat{\sigma}^2_{\varphi} &= \bm{h}^{\top} \left( \frac{1}{M - 1} \sum_{k = 1}^M \left(\mathbf{c}^k_{MAP} - \bar{\mathbf{c}} \right)\left(\mathbf{c}^k_{MAP} - \bar{\mathbf{c}} \right)^{\top} \right) \bm{h} \\
    &= \frac{1}{M - 1} \sum_{k = 1}^M \bm{h}^{\top} \left(\mathbf{c}^k_{MAP} - \bar{\mathbf{c}} \right)\left(\mathbf{c}^k_{MAP} - \bar{\mathbf{c}} \right)^{\top} \bm{h} \\
    &= \frac{1}{M - 1} \sum_{k = 1}^M \left[ \bm{h}^{\top} \left(\mathbf{c}^k_{MAP} - \bar{\mathbf{c}} \right) \right]^2 \\
    &= \frac{1}{M - 1} \sum_{k = 1}^M \left( \varphi_k - \bar{\varphi} \right)^2, \label{eq:post_func_variance}
\end{align}
where $\varphi_k = \bm{h}^{\top} \mathbf{c}^k_{MAP}$ and $\bar{\varphi} = \bm{h}^{\top} \bar{\mathbf{c}} = \frac{1}{M} \sum_{k = 1}^M \varphi_k$. The above algebra shows that the posterior variance of the functional can be computed using the functionals of the Monte Carlo samples without having to form the full empirical covariance matrix. See Algorithm~\ref{alg:posterior_covariance} for a succinct exposition of the above procedure. Note, the control flux can be built into the definition of $\bm{h}$ so that the functional has the desired units. Notice also that the functional does not need to be specified when creating the Monte Carlo ensemble. As long as the ensemble $\{\mathbf{c}_{MAP}^k\}_{k=1}^M$ is stored and made available to the end users, they may evaluate post-hoc the uncertainty of any functional that is of interest in their specific use-case.

\begin{algorithm}[t]
\caption{Monte Carlo Algorithm to Estimate Posterior Uncertainty in 4D-Var Data Assimilation}
    \medskip
    \textbf{Inputs}:
    \begin{itemize}
        \item $M \in \mathbb{N}$: Number of Monte Carlo samples.
        \item $\mathbf{R} \in \mathbb{R}^{n \times n}$: Observation error covariance built from retrieval uncertainties.
        \item $b^2 \in \mathbb{R}_+$: Scaling factor prior variance.
        \item $\mathbf{A} \in \mathbb{R}^{n \times m}$: Forward model mapping fluxes to X$_{CO_2}$ satellite observations (note, this can be known either explicitly via a matrix or implicitly via a computer simulator).
        \item $\bm{\mu} \in \mathbb{R}^m$: Control flux.
        \item $\bm{h} \in \mathbb{R}^m$: Vector defining the functional of interest.
    \end{itemize}
    \textbf{Steps}:
    \begin{enumerate}
        \item Let $\mathcal{S}$ denote an array of length $M$ that will store the MAP estimators for each Monte Carlo sample.
        \item For $k = 1, \dots, M$:
        \begin{enumerate}
            \item Simulate $\mathbf{c}_k \sim \mathcal{N}(\mathbf{1}, b^2 \mathbf{I}_m)$.
            \item Simulate $\bm{\epsilon}_k \sim \mathcal{N}(\mathbf{0}, \mathbf{R})$ and construct $\mathbf{y}_k = \mathbf{A} \bm{\mu} + \bm{\epsilon}_k$.
            \item Find MAP estimator $\mathbf{c}^k_{MAP}$. If $\mathbf{A}$ is known explicitly, it can be found using Equation~\eqref{eq:sf_map_est}. If $\mathbf{A}$ is known implicitly through a computational model, use a numerical optimizater (e.g., L-BFGS-B) to optimize a the 4D-Var cost function as defined in Equation~\eqref{eq:log_posterior}.
            \item $\mathcal{S}[k] \leftarrow \mathbf{c}^k_{MAP}$.
        \end{enumerate}
        \item Estimate Posterior Functional Variance:
        \begin{enumerate}
            \item Compute the mean Monte Carlo sample functional: $\bar{\varphi} = \frac{1}{M} \sum_{k = 1}^M \varphi_k$, where $\varphi_k = \bm{h}^\top \mathbf{c}^k_{MAP}$.
            \item Compute the empirical posterior functional variance:
            \begin{equation*}
                \hat{\sigma}^2_{\varphi} = \frac{1}{M - 1} \sum_{k = 1}^M \left( \varphi_k - \bar{\varphi} \right)^2
            \end{equation*}. \label{mc_emp_var}
        \end{enumerate}
    \end{enumerate}
\label{alg:posterior_covariance}
\end{algorithm}

\subsection{Quantifying the Monte Carlo Uncertainty} \label{subsec_monte_carlo_uncert}
Although Section \ref{subsec_monte_carlo_proc} establishes the equality of the Monte Carlo MAP estimator ensemble member covariance to the posterior covariance (and therefore the equality of the Monte Carlo ensemble member functional variance to the posterior functional variance), we have not yet established that the empirical covariance matrix (and functional variance) converges in probability to the true posterior covariance matrix (and functional variance). There are consistency results showing the empirical covariance matrix converging in probability to the true covariance matrix (see, for instance, Chapter 6 in \cite{wainwright}). However, since this application is primarily concerned with linear functionals of the form $\varphi(\mathbf{c}) = \bm{h}^{\top} \mathbf{c}$ as described in Section~\ref{subsec_monte_carlo_proc}, we can appeal directly to the consistency of the sample variance as shown, for example, in Chapter 5 of Casella and Berger (2002) \cite{casella_berger}.

Additionally, using the above algorithm, we would like to know either the uncertainty of the variance estimate given the number of Monte Carlo samples, or the number of samples required to obtain a particular level of Monte Carlo uncertainty on the variance. In essence, we would like to quantify the uncertainty of our uncertainty. To do so, we take a frequentist approach and construct confidence intervals on $\hat{\sigma}_\varphi^2$. The confidence intervals can be constructed by recognizing that the ratio of the Monte Carlo functional empirical posterior variance to the true functional posterior variance scaled by $(M-1)$ follows a $\chi^2_{M -1}$ distribution.

Since each sampled $\mathbf{c}^k_{MAP}$ is a linear function of other Gaussian samples (see Equation~\eqref{eq:sf_map_est}), $\mathbf{c}^k_{MAP}$ is also Gaussian, and thus the random variables $\varphi_k = \bm{h}^{\top} \mathbf{c}^k_{MAP}$ ($k = 1, \dots, M)$ are sampled independently and identically from a Gaussian distribution with some mean and variance $\sigma^2_{\varphi} = \bm{h}^\top \bm{\Sigma} \bm{h}$. By Theorem 5.3.1 of Casella and Berger (2002) \cite{casella_berger}, we have the following distributional result,
\begin{equation} \label{eq:samp_var_distrib}
    \frac{(M - 1) \hat{\sigma}^2_{\varphi}}{\sigma^2_{\varphi}} \sim \chi^2_{M - 1}.
\end{equation}
Thus, for $\alpha \in (0, 1)$, the distribution in Equation~\eqref{eq:samp_var_distrib} enables creating a $1 - \alpha$ confidence interval for the true posterior variance, $\sigma^2_{\varphi}$, as a function of the empirical posterior variance, $\hat{\sigma}^2_\varphi$. Using the exact distribution in Equation~\eqref{eq:samp_var_distrib}, we can create either one- or two-sided confidence intervals. Focusing on the two-sided case, we have,
\begin{equation} \label{eq:pivot_prob}
    \mathbb{P} \left\{ \chi^2_{M - 1, \alpha / 2} \leq \frac{(M - 1) \hat{\sigma}^2_{\varphi}}{\sigma^2_{\varphi}} \leq \chi^2_{M - 1, 1 - \alpha / 2} \right\} = 1 - \alpha,
\end{equation}
where $\chi^2_{M - 1, \alpha / 2}$ is the $\alpha / 2$-quantile of a chi-squared distribution with $M - 1$ degrees of freedom. Hence, with some algebraic manipulation we arrive at the confidence interval of the posterior variance,
\begin{equation} \label{eq:confid_int}
    \mathbb{P} \left\{\frac{(M - 1) \hat{\sigma}^2_{\varphi}}{\chi^2_{M - 1, 1 - \alpha / 2}} \leq \sigma^2_{\varphi} \leq  \frac{(M - 1) \hat{\sigma}^2_{\varphi}}{\chi^2_{M - 1, \alpha / 2}} \right\} = 1 - \alpha.
\end{equation}
Since in practice we would like to characterize uncertainty in the same units as the flux estimate, we can provide an analogous confidence interval for the posterior standard deviation by taking square roots of all the terms within the probability statement in Equation~\eqref{eq:confid_int}, giving
\begin{equation} \label{eq:confid_int_std}
   \mathbb{P} \left\{ \hat{\sigma}_\varphi \sqrt{\frac{M - 1}{\chi^2_{M - 1, 1 - \alpha / 2}}} \leq \sigma_{\varphi} \leq \hat{\sigma}_\varphi \sqrt{\frac{M - 1}{\chi^2_{M - 1, \alpha / 2}}} \right\} = 1 - \alpha.
\end{equation}
Equation~\eqref{eq:confid_int_std} facilitates the computation of a $(1 - \alpha) \times 100 \%$ frequentist interval estimator of the Bayesian credible interval for the functional of interest $\varphi$. For each endpoint of the true Bayesian credible interval, we find a confidence interval such that the probability that both endpoint confidence intervals simultaneously cover the true credible interval endpoints is $1 - \alpha$. Let $\gamma \in (0, 1)$ and $\varphi_{MAP} = \bm{h}^\top \mathbf{c}_{MAP}$ be the functional MAP estimator as described in Section~\ref{subsec_bayes_setup}. Because the posterior is Gaussian and $\varphi_{MAP}$ is a linear functional, it is a one-dimensional Gaussian. Hence, the Bayesian $(1 - \gamma) \times 100 \%$ credible interval is computed as follows, 
\begin{equation} \label{eq:bayes_post_int}
    \left[\barbelow{\varphi}^*, \bar{\varphi}^* \right] = \left[ \varphi_{MAP} - z_{1 - \gamma / 2} \cdot \sigma_\varphi, \varphi_{MAP} + z_{1 - \gamma / 2} \cdot \sigma_\varphi \right],
\end{equation}
where $z_{1 - \gamma / 2}$ is the $1 - \gamma / 2$ quantile of a standard Gaussian distribution. Equation~\eqref{eq:confid_int_std} allows us to construct the aforementioned endpoint confidence intervals as follows. For readability, define $L^2 := \frac{M - 1}{\chi^2_{M - 1, 1 - \alpha / 2}}$ and $R^2 := \frac{M - 1}{\chi^2_{M - 1, \alpha / 2}}$. Thus, we have the following,
\begin{align} \label{eq:confid_ints_cred_int}
   1 - \alpha &= \mathbb{P} \left\{ z_{1 - \gamma / 2} \hat{\sigma}_\varphi L \leq z_{1 - \gamma / 2} \sigma_{\varphi} \leq z_{1 - \gamma / 2} \hat{\sigma}_\varphi R \right\} \nonumber \\
   &= \mathbb{P} \{ - z_{1 - \gamma / 2} \hat{\sigma}_\varphi R \leq - z_{1 - \gamma / 2} \sigma_{\varphi} \leq - z_{1 - \gamma / 2} \hat{\sigma}_\varphi L \; \; \text{and} \nonumber  \\
   & \quad \quad \quad z_{1 - \gamma / 2} \hat{\sigma}_\varphi L \leq z_{1 - \gamma / 2} \sigma_{\varphi} \leq z_{1 - \gamma / 2} \hat{\sigma}_\varphi R \} \nonumber  \\
   &= \mathbb{P} \{\varphi_{MAP} - z_{1 - \gamma / 2} \hat{\sigma}_\varphi R \leq \varphi_{MAP} - z_{1 - \gamma / 2} \sigma_{\varphi} \leq \varphi_{MAP} - z_{1 - \gamma / 2} \hat{\sigma}_\varphi L \; \; \text{and} \nonumber  \\
   & \quad \quad \quad \varphi_{MAP} + z_{1 - \gamma / 2} \hat{\sigma}_\varphi L \leq \varphi_{MAP} + z_{1 - \gamma / 2} \sigma_{\varphi} \leq \varphi_{MAP} + z_{1 - \gamma / 2} \hat{\sigma}_\varphi R \} \nonumber  \\
   &= \mathbb{P} \{\varphi_{MAP} - z_{1 - \gamma / 2} \hat{\sigma}_\varphi R \leq \barbelow{\varphi}^* \leq \varphi_{MAP} - z_{1 - \gamma / 2} \hat{\sigma}_\varphi L \; \; \text{and} \nonumber  \\
   & \quad \quad \quad \varphi_{MAP} + z_{1 - \gamma / 2} \hat{\sigma}_\varphi L \leq \bar{\varphi}^* \leq \varphi_{MAP} + z_{1 - \gamma / 2} \hat{\sigma}_\varphi R \}. \nonumber 
\end{align}
More concisely, defining
\begin{align}
    \barbelow{I} &:= \left[\varphi_{MAP} - z_{1 - \gamma / 2} \hat{\sigma}_\varphi R, \varphi_{MAP} - z_{1 - \gamma / 2} \hat{\sigma}_\varphi L \right], \\
    \bar{I} &:= \left [\varphi_{MAP} + z_{1 - \gamma / 2} \hat{\sigma}_\varphi L, \varphi_{MAP} + z_{1 - \gamma / 2} \hat{\sigma}_\varphi R \right],
\end{align}
it follows that
\begin{equation}
    \mathbb{P} \left\{\barbelow{\varphi}^* \in \barbelow{I} \; \; \text{and} \; \; \bar{\varphi}^* \in \bar{I} \right\} = 1 - \alpha.
\end{equation}
The intervals $\barbelow{I}$ and $\bar{I}$ quantify uncertainty on uncertainty, and provide a rigorous probabilistic characterization of the Monte Carlo procedure's uncertainty.

In practice, the original Bayesian credible interval in Equation~\eqref{eq:bayes_post_int} can be modified to account for the Monte Carlo uncertainty. To obtain an \emph{upper bound} on the Bayesian credible interval, we apply an \emph{inflation} factor (defined above as $R$), and thus obtain the interval $\left[\barbelow{\varphi}_u, \bar{\varphi}_u \right] = \left[ \varphi_{MAP} - z_{1 - \gamma / 2} \hat{\sigma}_\varphi R, \varphi_{MAP} + z_{1 - \gamma / 2} \hat{\sigma}_\varphi R \right]$, such that $\mathbb{P} \left\{ \left[\barbelow{\varphi}_u, \bar{\varphi}_u \right]  \supset \left[\barbelow{\varphi}^*, \bar{\varphi}^* \right] \right\} = 1 - \alpha / 2$. This probability is $(1 - \alpha / 2)$ instead of $(1 - \alpha)$ since the probability in Equation~\eqref{eq:pivot_prob} evaluated with only the lower bound yields a probability of $(1 - \alpha / 2)$. Following the same steps as above, we obtain lower and upper bounds on the lower and upper endpoints of credible interval~\eqref{eq:bayes_post_int}, respectively, holding with probability exactly $(1 - \alpha / 2)$. Similarly, to obtain a \emph{lower bound} on the Bayesian credible interval, we apply a \emph{deflation} factor (defined above as $L$) and thus obtain the interval $\left[\barbelow{\varphi}_l, \bar{\varphi}_l \right] = \left[ \varphi_{MAP} - z_{1 - \gamma / 2} \hat{\sigma}_\varphi L, \varphi_{MAP} + z_{1 - \gamma / 2} \hat{\sigma}_\varphi L \right]$ such that $\mathbb{P} \left\{ \left[\barbelow{\varphi}_l, \bar{\varphi}_l \right]  \subset \left[\barbelow{\varphi}^*, \bar{\varphi}^* \right] \right\} = 1 - \alpha / 2$, holding with equality by the same logic as that used for the inflation factor.

Observing that the aforementioned inflation and deflation factors monotonically asymptote to one as the number of Monte Carlo samples $M$ gets large, the effect of the Monte Carlo procedure wanes as the number of samples grows. As shown in Table~\ref{table:ratio_vals}, the deflation factor monotonically approaches one from below as $M$ gets large while the inflation factor monotonically approaches one from above as $M$ gets large. As is characteristic of DA methods, each Monte Carlo iteration requires a non-trivial amount of computation, which practically restricts the number of Monte Carlo samples that can be obtained. As such, the inflated interval protects against underestimating the uncertainty, while the deflated interval provides a lower bound or ``best-case'' scenario for the uncertainty of the Bayesian procedure.

\begin{table}[ht]
    \centering
    \caption{Inflation and deflation factors for Monte Carlo (MC) estimated posterior standard deviation with $\alpha = 0.05$. When $M = 100$, by inflating the MC estimated posterior standard deviation by a factor of $1.1607$ (inflating by $16.07\%$), the extra uncertainty resulting from the MC procedure is accounted for with $97.5\%$ confidence. Similarly, when $M = 100$, deflating the MC estimated posterior standard deviation by a factor of $0.8785$ provides a lower bound on the true underlying Bayesian uncertainty with $97.5\%$ confidence. When considered simultaneously, the inflation and deflation factors bracket the true uncertainty with $95\%$ confidence.}
    \label{table:ratio_vals}
\begin{tabular}{l l l}
        \hline
        \# Monte Carlo samples, $M$ & Deflation: $L = \sqrt{\frac{M - 1}{\chi^2_{M - 1, 1 - \alpha / 2}}}$ & Inflation: $R = \sqrt{\frac{M - 1}{\chi^2_{M - 1, \alpha / 2}}}$ \\ \hline
        10 & 0.6987 & 1.7549 \\
        100 & 0.8785 & 1.1607 \\
        1,000 & 0.9580 & 1.0458 \\
        10,000 & 0.9863 & 1.0141 \\
        100,000 & 0.9956 & 1.0044 \\
        1,000,000 & 0.9986 & 1.0014 \\
        \hline
    \end{tabular}
\end{table}

\section{Numerical Examples} \label{sec:experimental_examples}
\subsection{Low-Dimensional Example} \label{subsec:low_dim_ex}
We construct a two-dimensional toy example to provide a numerical demonstration that this MC procedure computes a consistent estimate of the posterior covariance, and is numerically close in practice. We define a linear forward model with the following matrix:
\begin{equation} \label{eq:defA}
\mathbf{A} = 
\begin{bmatrix}
    1 - \epsilon & \epsilon \\
    \epsilon & 1 - \epsilon
\end{bmatrix},
\end{equation}
where $\epsilon > 0$. Let $\bm{\theta} \in \mathbb{R}^2$ be the true state of some physical quantity and $\bm{\mu} \in \mathbb{R}^2$ be the control state. We use the values in Table~\ref{table:parameter_settings} to demonstrate the agreement between the analytical equations for the Bayesian procedure and the Monte Carlo procedure in addition to showing their agreement with the empirical covariance computed from the Monte Carlo ensemble members.
\begin{table}[t]
    \centering
    \caption{Parameter settings for the low-dimensional example.}
    \begin{adjustbox}{max width=\textwidth,center}
    \begin{tabular}{l l l}
        \hline
        Parameter & Value & Description \\
        \hline \hline
        $\mathbf{\theta}$ & $\begin{bmatrix} 1 & 2 \end{bmatrix}^{\top}$ & True state of the system \\
        $\bm{\mu}$ & $\begin{bmatrix} \frac{1}{2} & 1 \end{bmatrix}^{\top}$ & Control state \\
        $n$ & $2$ & The number of observations in each ensemble member (also the dimension of $\mathbf{y}$ in Equation~\eqref{eq:like_y}) \\
        $M$ & $10^6$ & The number of MC ensemble members \\
        $b^2$ & $4$ & The prior variance for each element in the scaling factor vector \\
        $\epsilon$ & $0.05$ & Parameter of the matrix $\mathbf{A}$ defined in Equation~\eqref{eq:defA} \\
        $\sigma^2$ & $1$ & Observation error variance \\
        \hline
    \end{tabular}
    \end{adjustbox}
    \label{table:parameter_settings}
\end{table}
Using the Table~\ref{table:parameter_settings} parameters and Equation~\eqref{eq:sf_posterior_cov}, we obtain the following posterior covariance matrix:
\begin{equation}
\bm{\Sigma} = 
\begin{bmatrix}
    2.10838562 & -0.0867085 \\
    -0.0867085 & 0.8693668
\end{bmatrix}.
\end{equation}
\noindent For the analytical covariance of the MAP estimator, we obtain the following matrix using Equation~\eqref{eq:map_covariance}:
\begin{equation}
\bm{\Sigma}_{\mathbf{c}_{MAP}^k}  = 
\begin{bmatrix}
    2.10838562 & -0.0867085 \\
    -0.0867085 & 0.8693668
\end{bmatrix}.
\end{equation}
Indeed, these matrices are expected to be the same. Using simulated ensemble members, we obtain the following empirical covariance matrix using Equation~\eqref{eq:empirical_cov}:
\begin{equation}
\widehat{\bm{\Sigma}}  = 
\begin{bmatrix}
    2.09697751 & -0.08748562 \\
    -0.08748562 & 0.87053267
\end{bmatrix}.
\end{equation}
This empirical covariance matrix is numerically very close to the analytical matrices. Using Equation (6.12) from \cite{wainwright}, we can compute an upper bound on the relative error between the empirical covariance matrix and the true posterior covariance matrix with respect to the $\ell_2$ operator norm that holds with at least some desired probability. Concretely, with $n = 2$ and $M = 10^6$, this deviation bound implies that the relative error is at most $0.01784$ ($1.784\%$) with probability at least $95\%$. This tight error bound matches the above numerical closeness.

All of the above quantities exist in the scaling factor space. For the physical quantity space covariance matrix, we obtain the following from Equations \eqref{eq:flux_posterior_cov} and \eqref{eq:mc_cov_phys_quant}:
\begin{equation}
    \bm{\Gamma} = \begin{bmatrix}
    0.52709641 & -0.04335425 \\
    -0.04335425 & 0.8693668
\end{bmatrix}
\end{equation}
\begin{equation}
    \text{Cov}[\bm{\theta}_k] = \begin{bmatrix}
    0.52709641 & -0.04335425 \\
    -0.04335425 & 0.8693668
\end{bmatrix}
\end{equation}
\begin{equation}
\frac{1}{M - 1} \sum_{k=1}^M (\theta_k - \bar{\theta}) (\theta_k - \bar{\theta})^{\top}  = 
\begin{bmatrix}
   0.52424438 & -0.04374281 \\
    -0.04374281 & 0.87053267
\end{bmatrix}
\end{equation}
Again, we see that the agreement between the analytical forms and the empirical covariance matrix is very close. 

This toy example demonstrates the correctness of the Monte Carlo procedure in terms of the covariance matrix. However, as seen in Algorithm~\ref{alg:posterior_covariance}, we are actually interested in scenarios in which we are considering a functional of high-dimensional parameters, and hence a variance. Such an example is considered in the following section.

\subsection{Carbon Flux Inversion OSSE} \label{sec:osse_example}
We show here an example of this Monte Carlo procedure being used to compute posterior uncertainties for global carbon fluxes. We follow the flux inversion setup used by Byrne et al.\ \cite{byrne_gosat} (see Section 2.3 of that study). This setup uses the GEOS-Chem Adjoint model \cite{Henze-2007} to estimate scaling factors on a $4^{\circ} \times 5^{\circ}$ surface grid from January 2010 up to and including August. For each spatial point, there is one scaling factor parameter for each month, totaling $m = 72 \times 46 \times 8 = 26,492$ scaling factors parameters. This model is linear in terms of realistic fluxes (e.g., not including abnormally large negative fluxes), and hence amenable to this uncertainty quantification procedure. The OSSE defines ground-truth fluxes from the Joint UK Land Environment Simulator (JULES) \cite{clark2011,harper2018} and uses Net Ecosystem Exchange (NEE) fluxes from NOAA's CarbonTracker version CT2016 (\cite{Peters-2007}, with updates documented at \url{https://www.esrl.noaa.gov/gmd/ccgg/carbontracker/}) as the control fluxes. The satellite $X_{CO_2}$ observations for the assimilation are generated from the JULES fluxes by running a forward GEOS-Chem simulation and sampling the model with the GOSAT observational coverage and observation operator \cite{odell}.

The prior uncertainty, as described in Equation \eqref{eq:prior_c}, is set to $b = 1.5$ (where $\mathbf{B} := b^2 \mathbf{I}_M$). To perform the Monte Carlo procedure, we draw $M = 60$ ensemble members, as described in Sec.~\ref{subsec_monte_carlo_proc}. The $X_{CO_2}$ observation uncertainty $\bm{\Sigma}$ (a diagonal matrix, as the observations are assumed to be independent) comes directly from the GOSAT data product and varies between observations. For each ensemble member $k$, the output of the GEOS-Chem Adjoint optimization provides monthly scaling factor MAP estimators $\mathbf{c}_{MAP}^k$ according to the ensemble member inputs as described by Equation~\eqref{eq:sf_map_est}. Each ensemble MAP estimator is then multiplied by the control flux to obtain a MAP estimator in flux space.

The functionals of interest $\varphi$ are monthly global fluxes. The flux values on the $3$-hour $4^{\circ} \times 5^{\circ}$ spatial-temporal grid are mapped to a global monthly flux using a weighted average with weights proportional to the surface area of each grid cell and uniform time weighting. The global flux posterior variance is computed for each month by finding the empirical variance of the Monte Carlo global flux members, as shown in Equation~\eqref{eq:post_func_variance}. To get a sense of how the DA is reducing prior uncertainty, for each month, we compute a $\%$ uncertainty reduction as follows:
\begin{equation} \label{eq:uq_reduc}
    \text{\% Uncertainty Reduction} = 1 - \frac{\sigma_{posterior}}{\sigma_{prior}}.
\end{equation}
Since we do not precisely know the posterior standard deviation in Equation~\eqref{eq:uq_reduc}, we consider the reduction both in terms of the raw Monte Carlo point estimate of the posterior standard deviation and its inflated version (i.e., $R$ as defined in Section~\ref{subsec_monte_carlo_uncert}).

The left side of Figure~\ref{fig:global_flux_results_percent_reduc} shows the timeseries of global mean functionals and their credible intervals.
\begin{figure}[htp]
    \centering
    \includegraphics[width=\textwidth]{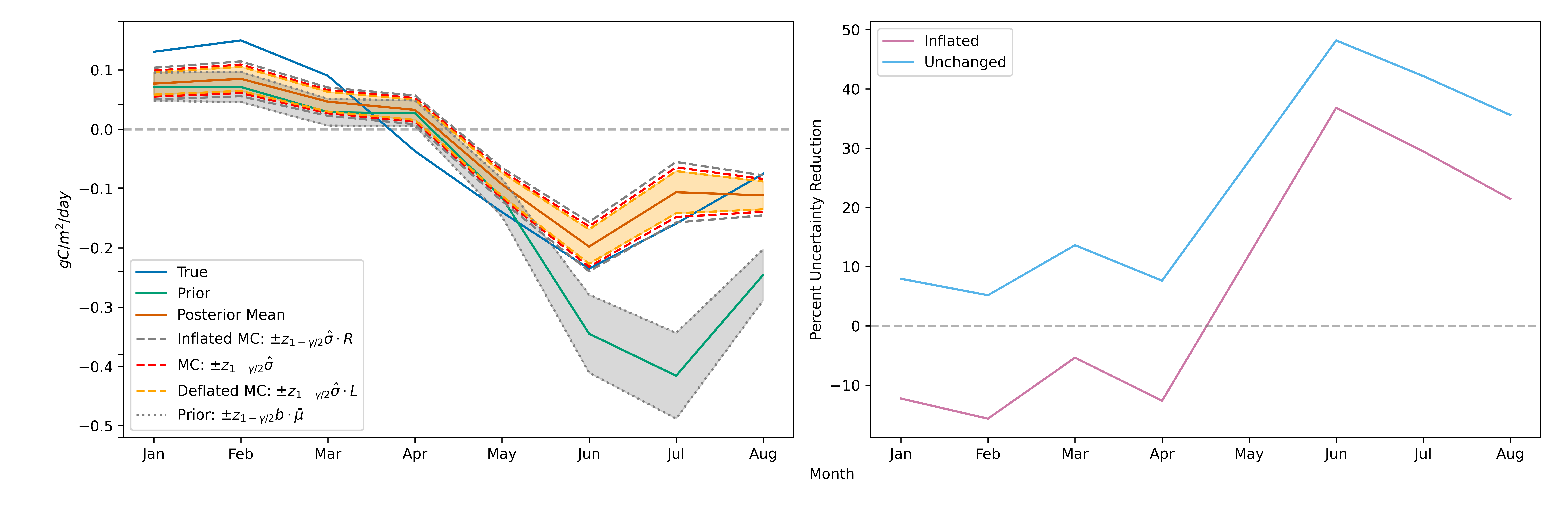}
    \caption{\textbf{(Left)} Estimated posterior $(1 - \gamma) \times 100\% = 95\%$ credible intervals around the monthly global flux functionals show markedly improved uncertainty over the prior during boreal summer months. The three interval types shown are the unchanged MC estimated intervals (red), the inflated MC estimated intervals (gray), and the deflated MC estimated intervals (orange). As described in Section~\ref{subsec_monte_carlo_uncert}, for each month, the true upper and lower credible interval endpoints are contained within the inflated and deflated endpoints with probability $1 - \alpha = 0.95$.  Note, $\hat{\sigma}$ is a shorthand notation for the empirical functional standard deviation as defined in Step~\eqref{mc_emp_var} of Algorithm~\ref{alg:posterior_covariance}, $\bar{\mu}$ is the globally averaged control flux to which the prior uncertainty, $b$, is applied for each month, and $R$ and $L$ are the inflation and deflation factors, respectively, as defined in Table~\ref{table:ratio_vals}. We observe that with even as few as $M = 60$ ensemble members, at the monthly/global scale, the magnitude of the Monte Carlo sampling uncertainty is small in comparison to the posterior uncertainty. \textbf{(Right)} Percent reduction in uncertainty from prior to posterior for the monthly global fluxes is most significant during the boreal summer. The light blue curve shows the percent reduction estimated with the unchanged MC estimated posterior standard deviation, while the magenta curve shows the percent reduction estimated with the inflated MC estimated posterior standard deviation. The true reduction is larger than the reduction shown by the latter curve with $97.5\%$ confidence.}
    \label{fig:global_flux_results_percent_reduc}
\end{figure}
The posterior flux is shown to have reduced error against the true flux, especially during the boreal summer months. Similarly, the Monte Carlo posterior uncertainty estimate shows considerable reduction relative to the prior. The uncertainty estimates with inflated endpoints, increase the posterior uncertainty by $22\%$ while the deflated endpoints decrease the posterior uncertainty by $15\%$, resulting in credible interval endpoint bounds that capture the true credible interval endpoints with $95\%$ probability. The right side of Figure~\ref{fig:global_flux_results_percent_reduc} further emphasizes the prior to posterior uncertainty reduction that we mathematically expect. However, we notice that the inflated uncertainty is only reduced during boreal summer months. In January, February, March, and April the inflated Monte Carlo estimated posterior uncertainty is actually larger than the prior uncertainty. There is a logical explanation for this: since most of the landmass generating NEE fluxes is in the Northern Hemisphere and GOSAT requires sunlight to measure $X_{CO_2}$, the satellite observations impose much weaker constraints on the fluxes during boreal winter. Furthermore, since the prior uncertainty is defined as a percentage, the prior is more concentrated during the boreal winter months when the absolute magnitude of the CarbonTracker fluxes is smaller. As a result of these two effects, the actual posterior uncertainty during the winter months is only slightly smaller than the prior uncertainty. Since we are obtaining a noisy Monte Carlo estimate of this uncertainty from using $60$ ensemble members, the inflated value accounting for the Monte Carlo uncertainty of the posterior uncertainty is slightly larger than the prior uncertainty.

\section{Conclusion} \label{sec:conclusion}
For Bayesian uncertainty quantification in which the forward model is only available as a simulator, the carbon flux estimation community has proposed a useful Monte Carlo method to compute posterior uncertainties. This method is especially well-suited to DA tasks since it is parallelizable, works with computationally intensive physical simulators, and allows for flexible post-hoc uncertainty quantification on any desired functional of the model parameters. In this note, we analytically established the mathematical correctness of this procedure in the case of Gaussian prior and error distributions and provided additional uncertainty quantification to account for the Monte Carlo sampling variability in the final estimated credible interval. We also provided two numerical examples. In the first, we demonstrated the agreement between the analytical equations and empirical results for an explicitly known linear forward model. In the second, we showed that this procedure applies to a large-scale DA problem in the form of a carbon flux inversion OSSE, and reasoned that the uncertainty quantification results are mathematically and practically sensible. Future investigations of this method could be based on an exploration of how many ensemble members must be sampled before the Monte Carlo uncertainty is sufficiently less than the posterior uncertainty.

\section*{Acknowledgments}

This work was supported by NSF grant DMS-2053804, JPL RSA No. 1670375 and 1689177 and a grant from the C3.AI Digital Transformation Institute. Part of this research was carried out at the Jet Propulsion Laboratory, California Institute of Technology, under a contract with the National Aeronautics and Space Administration (grant no. 80NM0018D004). Liu and Byrne would like to acknowledge the funding support from NASA Orbiting Carbon Observatory Science Team program (17-OCO2-17-0013). Finally, we would like to thank Anna Harper for providing the JULES fluxes used in this study's OSSE, the STAMPS research group at Carnegie Mellon University for supporting this work, and the Uncertainty Quantification group at the Jet Propulsion Laboratory for facilitating this collaboration. We would like to acknowledge high-performance computing support from Cheyenne (doi:10.5065/D6RX99HX) provided by NCAR's Computational and Information Systems Laboratory, sponsored by the National Science Foundation.

\appendix

\section{Supporting Algebraic Results} \label{sec:appendix_multivariate_algebra}
There are a few key properties of the element-wise multiplication operation that must be stated in order to support the derivation of the equations presented in this paper.

For the following, let $\mathbf{x} \in \mathbb{R}^m$ be a random vector such that $\mathbb{E}[\mathbf{x}] = \bm{\mu}$ and $\text{Cov}[\mathbf{x}] = \bm{\Sigma}$. Additionally, suppose $\mathbf{a} \in \mathbb{R}^m$ and $\mathbf{A} \in \mathbb{R}^{n \times m}$. 

\begin{lemma} \label{lemma:E_circ}
$\mathbb{E}[\mathbf{x} \circ \mathbf{a}] = \mathbb{E}[\mathbf{x}] \circ \mathbf{a}$\,.
\end{lemma}
\begin{proof}
By definition, we have
\begin{equation}
    \mathbb{E}[\mathbf{x} \circ \mathbf{a}] = [\mathbb{E}[x_i a_i]]_i = [a_i \mathbb{E}[x_i]]_i = \mathbb{E}[\mathbf{x}] \circ \mathbf{a}\,.
\end{equation}
\end{proof}
\begin{lemma} \label{lemma:cov_circ}
    $\mathrm{Cov}[\mathbf{x} \circ \mathbf{a}] = \mathrm{Cov}[\mathbf{x}] \circ \mathbf{a} \mathbf{a}^{\top}$
\end{lemma}
\begin{proof}
There are two terms that need to be computed: (1) $\text{Var}[x_i a_i]$ and (2) $\text{Cov}[x_i a_i, x_j a_j]$. (1) is straightforward by properties of variance, namely, $\text{Var}[x_i a_i] = a_i^2 \text{Var}[x_i]$. (2) simply requires the definition of covariance, i.e.,
\begin{equation}
\text{Cov}[x_i a_i, x_j a_j] = \mathbb{E} \big[ \big( x_i a_i - \mathbb{E}[x_i a_i] \big) \big( x_j a_j - \mathbb{E}[x_j a_j] \big) \big] = a_i a_j \text{Cov}[x_i, x_j].
\end{equation}
Hence, it follows that $\text{Cov}[\mathbf{x} \circ \mathbf{a}] = \text{Cov}[\mathbf{x}] \circ \mathbf{a} \mathbf{a}^{\top}$.
\end{proof}
\begin{lemma} \label{lem:matrix_decomp}
Let $\bm{\mu} \in \mathbb{R}^m$ and let $\mathbf{A}_{\bm{\mu}}$ be such that the $i$th row of $\mathbf{A}_{\bm{\mu}}$ is equal to $[A_{ij} \mu_j]_j$. Then the following equation holds:
\begin{equation}
    \mathbf{A}_{\mu}^{\top} \mathbf{A}_{\mu} = \mathbf{A}^{\top} \mathbf{A} \circ \bm{\mu} \bm{\mu}^{\top}. \label{eq:mat_prod}
\end{equation}
\end{lemma}
\begin{proof}
To prove Equation \eqref{eq:mat_prod}, first note that $\bm{\mu} \bm{\mu}^{\top} = \left[ \mu_i \mu_j\right]_{ij}$. Let $i \in [m]$ and $j \in [n]$. By definition, we have $$ \left[ \mathbf{A}_{\bm{\mu}}^{\top} \mathbf{A}_{\bm{\mu}} \right]_{ij} = \sum_{l = 1}^m A_{l i} A_{l j} \mu_i \mu_j = \mu_i \mu_j \sum_{l = 1}^m A_{l i} A_{l j}.$$ Hence, we can see that $$ \left[ \mathbf{A}_{\bm{\mu}}^{\top} \mathbf{A}_{\bm{\mu}} \right]_{ij} = \left[ \sum_{l = 1}^m A_{l i} A _{l j}  \right]_{ij} \cdot \left[ \mu_i \mu_j \right]_{ij} $$ and we have the desired result.
\end{proof}
\begin{corollary} \label{cor:psd}
    Let $\mathbf{M} \in \mathbb{R}^{n \times n}$ be a positive-definite matrix and let $\mathbf{A}_{\bm{\mu}}$ be defined as in Lemma~\ref{lem:matrix_decomp}. It follows that,
    \begin{equation}
        \mathbf{A}_{\bm{\mu}}^\top \mathbf{M} \mathbf{A}_{\bm{\mu}} = \mathbf{A}^\top \mathbf{M} \mathbf{A} \circ \bm{\mu} \bm{\mu}^\top.
    \end{equation}
\end{corollary}
\begin{proof}
    Since $\mathbf{M}$ is positive-definite, it has a lower-triangular Cholesky decomposition, $\mathbf{M} = \mathbf{L} \mathbf{L}^\top$. For all $i \in [m]$ and $j \in [n]$, we have the following equivalence:
    \begin{equation}
        \left[ \mathbf{L}^\top \mathbf{A}_{\bm{\mu}}\right]_{ij} = \sum_{l = 1}^n A_{lj} L_{li} \mu_j = \left[ \left( \mathbf{L}^\top \mathbf{A}\right)_{\bm{\mu}}\right]_{ij}.
    \end{equation}
    Therefore, $\mathbf{L}^\top \mathbf{A}_{\bm{\mu}} = \left(\mathbf{L}^\top \mathbf{A} \right)_{\bm{\mu}}$. Thus, Lemma~\ref{lem:matrix_decomp} implies,
    \begin{equation}
        \mathbf{A}_{\bm{\mu}}^\top \mathbf{M} \mathbf{A}_{\bm{\mu}} = \left( \mathbf{L}^\top \mathbf{A}_{\bm{\mu}} \right)^\top \left( \mathbf{L}^\top \mathbf{A}_{\bm{\mu}} \right) = \left(\mathbf{L}^\top \mathbf{A} \right)_{\bm{\mu}}^\top \left(\mathbf{L}^\top \mathbf{A} \right)_{\bm{\mu}} = \left(\mathbf{L}^\top \mathbf{A} \right)^\top \left(\mathbf{L}^\top \mathbf{A} \right) \circ \bm{\mu} \bm{\mu}^\top,
    \end{equation}
    and the result follows since $\left(\mathbf{L}^\top \mathbf{A} \right)^\top \left(\mathbf{L}^\top \mathbf{A} \right) = \mathbf{A}^\top \mathbf{M} \mathbf{A}$.
\end{proof}
\begin{lemma} \label{lem:circ_through_matmul}
   Let $\mathbf{A}_{\bm{\mu}} \in \mathbb{R}^{n \times m}$ be defined as in the above, $\mathbf{M} \in \mathbb{R}^{n \times n}$ and $\mathbf{y} \in \mathbb{R}^n$. Then,
   \begin{equation}
       \mathbf{A}_{\bm{\mu}}^\top \mathbf{M} \mathbf{y} = \mathbf{A}^\top \mathbf{M} \mathbf{y} \circ \bm{\mu}.
   \end{equation}
\end{lemma}
\begin{proof}
    It is sufficient to show the case when $\mathbf{M} = \mathbf{I}_n$, as otherwise we can simply define a new vector $\tilde{\mathbf{y}} = \mathbf{M} \mathbf{y}$. By matrix multiplication, for all $j \in [m]$,
    \begin{equation}
        \left[\mathbf{A}_{\bm{\mu}}^\top \mathbf{y} \right]_j = \sum_{i = 1}^n A_{ij} y_i \mu_j = \left[\mathbf{A}^\top \mathbf{y} \circ \bm{\mu} \right]_j.
    \end{equation}
\end{proof}
\begin{lemma} \label{circ_lemma}
Let $\mathbf{A}_{\bm{\mu}} \in \mathbb{R}^{n \times m}$ be defined as in the above. Then,
\begin{equation}
    \mathbf{A}(\mathbf{a} \circ \bm{\mu}) = \mathbf{A}_{\bm{\mu}} \mathbf{a}\,.
\end{equation}
\end{lemma}
\begin{proof}
This property follows simply from the definition of $\mathbf{A}_{\bm{\mu}}$ and matrix multiplication.
\end{proof}
\begin{corollary}
    Define $\mathbf{A} \in \mathbb{R}^{n \times m}$ as above, let $\alpha, \beta \in \mathbb{R}$, and let $\mathbf{a}, \mathbf{b}, \bm{\mu} \in \mathbb{R}^m$. Then $\mathbf{A}(\mathbf{a} \circ \mathbf{\bm{\mu}})$ is linear in $\mathbf{a}$.
\end{corollary}
\begin{proof}
    By Lemma~\ref{circ_lemma}, we have,
    \begin{equation}
        \mathbf{A} \left( \left(\alpha \mathbf{a} + \beta \mathbf{b} \right) \circ \bm{\mu} \right) = \mathbf{A}_{\bm{\mu}} \left(\alpha \mathbf{a} + \beta \mathbf{b} \right) = \alpha \mathbf{A}(\mathbf{a} \circ \bm{\mu}) +  \beta \mathbf{A} (\mathbf{b} \circ \bm{\mu}).
    \end{equation}
\end{proof}

\bibliographystyle{acm}
\bibliography{references}
\end{document}